\documentclass[10pt,A4paper,conference]{IEEEtran}
\usepackage{amsmath}
\usepackage{amssymb}
\usepackage{upref}
\usepackage{amsfonts}
\usepackage{graphicx}
\usepackage{epic}
\usepackage{eepic}
\usepackage{psfrag}
\usepackage{verbatim}

\newcommand{\cA}{{\cal A}}

\newcommand{\cC}{{\cal C}}
\newcommand{\cG}{{\cal G}}
\newcommand{\cH}{{\cal H}}

\newcommand{\cP}{{\cal P}}

\newcommand{\sP}{\cP}
\newcommand{\sG}{\cG}

\newcommand{\Gr}{\smash{{\sG\kern-1.5pt}_q\kern-0.5pt(n,k)}}
\newcommand{\Grtwo}{\smash{{\sG\kern-1.5pt}_2\kern-0.5pt(n,k)}}
\newcommand{\Gkone}{\smash{{\sG\kern-1.5pt}_q\kern-0.5pt(n,k_1)}}
\newcommand{\Gktwo}{\smash{{\sG\kern-1.5pt}_q\kern-0.5pt(n,k_2)}}
\newcommand{\Ps}{\smash{{\sP\kern-2.0pt}_q\kern-0.5pt(n)}}

\newcommand{\bI}{{\bf I}}

\newcommand{\bi}{{\bf i}}
\newcommand{\bv}{{\bf v}}

\newtheorem{theorem}{Theorem}

\newtheorem{remark}{Remark}

\begin{document}

\title{High Dimensional Error-Correcting Codes}

\author{\authorblockN{Eitan Yaakobi}
\authorblockA{Dept. of Electrical and Computer Engineering\\
University of California, San Diego\\
La Jolla, CA 92093, USA \\
Email: eyaakobi@ucsd.edu} \and
\authorblockN{Tuvi Etzion}
\authorblockA{Dept. of Computer Science\\
Technion-Israel Institute of Technology\\
Haifa 32000, Israel \\
Email: etzion@cs.technion.ac.il}}

\maketitle
\begin{abstract}
In this paper we construct multidimensional codes with high
dimension. The codes can correct high dimensional errors which
have the form of either small clusters, or confined to an area
with a small radius. We also consider small number of errors in a
small area. The clusters which are discussed are mainly spheres
such as semi-crosses and crosses. Also considered are clusters
with small number of errors such as 2-bursts, two errors in
various clusters, and three errors on a line. Our main focus is on
the redundancy of the codes when the most dominant parameter is
the dimension of the code.
\end{abstract}

\section{Introduction}
\label{sec:introduction}

Multidimensional coding in general and two-dimensional coding in
particular is a subject which attracts a lot of attention in the
last three decades. But, although the related theory of the
one-dimensional case is well developed, the theory for the
multidimensional case is developed rather slowly. This is due to
the fact that most of the one-dimensional techniques are not
generalized easily to higher dimensions and usually many different
techniques are used in the multidimensional case.

\begin{remark}
In our discussion we will consider noncyclic arrays, even if the
construction works on a cyclic array, i.e., a torus. This is done
for convenient reasons. In the following redundancy definition,
the array is considered to be cyclic. But, since the size of the
array is very large we will omit the minor difference in the
redundancy between a cyclic array and a noncyclic array.
\end{remark}

A binary multidimensional error-correcting code corrects errors
which occur in a multidimensional array. Throughout the paper the
volume of the array is $N$. If we are given a set with $\beta$
possible patterns of errors (no error is also such a pattern) that
can occur anywhere in the array then the redundancy of the code
must satisfy $r \geq \text{log}~ (N \cdot \beta ) = \text{log}~ N+
\text{log}~ \beta$ (all logarithms in this paper are in base 2).
The difference $r- \text{log}~ N$ is called the {\it excess
redundancy} of the code~\cite{Abd86,AMT}.

Abdel-Ghaffar~\cite{Abd86} constructed binary two-dimensional
codes which correct a cluster of a rectangle shape with area $B$
for which $r = \lceil \text{log}~N \rceil +B$. These codes attain
the lower bound on the excess redundancy. There is no known
generalization for the construction of Abdel-Ghaffar~\cite{Abd86}
to more than two dimensions. Moreover, the number of length
parameters on which the construction works is very limited. A
construction in~\cite{EtYa} produces a $D$-dimensional code for
correction of a $D$-dimensional box-error with redundancy $\lceil
\text{log}~ N \rceil +B+\lceil \text{log}~ b_1 \rceil$, where
$b_1$ is the length of the $D$-dimensional box in the first
dimension. For the two-dimensional case, this construction is more
flexible in its parameters than the construction in~\cite{Abd86}.

In this paper we are interested in $D$-dimensional codes, where $D
>2$ is usually very large. On the other hand, we are
interested either in a small number of errors or that the cluster
is spread in radius at most one from the center of the error
event.

How can we correct such bursts? If the size of the burst is one
then we can always use an one-dimensional Hamming code folded on
the $D$-dimensional array to obtain an optimal code. Given a set
$S$ with patterns of errors, the most natural and simple way to
correct an error from $S$ is to correct a box-error which contains
all possible errors from $S$. This might result in a large excess
redundancy as the number of patterns in $S$ might be much smaller
than the number of patterns defined by a box-error. The goal of
this paper is to construct codes with excess redundancy much
smaller than the one implied by a correction of a related
box-error. If the size of the burst is two then we will see in the
sequel that a code with optimal, or almost optimal, excess
redundancy can be constructed. But, if the size of the burst is
three then we don't know how to construct a code which attains the
lower bound on the excess redundancy. In fact, the question how to
construct an optimal code which corrects an arbitrary cluster is
still open.

The rest of the paper is organized as follows. In
Section~\ref{sec:basic} we present the definitions of linear codes
and shapes which are discussed throughout this paper. In
Section~\ref{sec:optimal} we discuss codes for which the error is
small and confined to an one-dimensional line. We will examine two
types of errors, 2-bursts and 3-bursts (bursts of length two and
three, respectively) on a line, where a $b$-{\it burst} is any set
of errors that is confined to an area of size $b$. In
Section~\ref{sec:coloring} we discuss a coloring method presented
in~\cite{EtYa} and explain how it is designed to correct cluster
errors. We show how to correct error whose shape is a semi-cross
(corner) with arms of length one (radius of length one) or a cross
(Stein's sphere) with arms of length one (radius of length one).
These two shapes will exhibit an excellent example for the
strength of the coloring method. In Section~\ref{sec:limit_weight}
we consider clusters with small weight inside a relatively larger
cluster. We present asymptotically optimal solutions for the case
where the weight is two and the cluster is a semi-cross, cross, or
a two-dimensional square. In Section~\ref{sec:conclude} we
conclude and present the goals for the future research.
\section{Basic Definitions}
\label{sec:basic}

A binary multidimensional $b$-error-correcting code is a set $C$
consisting of $D$-dimensional binary arrays of the same size, such
that if we are given an array $\cA$ from $C$ and the values of up
to $b$ positions in $\cA$ are changed, then we will be able to
recover $\cA$. We consider only linear codes as done in all
previous works.

A binary $D$-dimensional linear code $C$ is a linear subspace of
the $n_1\times n_2 \times \cdots \times n_D$ binary matrices. If
the subspace is of dimension $N-r$, where $N=\prod_{\ell=1}^D
n_\ell$, we say that the code is an $[n_1\times n_2 \times \cdots
\times n_D , N -r]$ code. The code can be also defined by its
parity-check matrix. Let $H=(h_{\bi , j})$, where $\bi \in \bI$,
$\bI =\{ (i_1,i_2,\ldots,i_D) ~:~ 0 \leq i_\ell \leq n_\ell-1 \}$,
and $0 \leq j \leq r-1$, be a $(D+1)$-dimensional binary matrix of
size $n_1\times n_2 \times \cdots \times n_D \times r$, consisting
of $r$ linearly independent $n_1\times n_2 \times \cdots \times
n_D$ matrices. Let $c=(c_{\bi})$ denote a binary $n_1\times n_2
\times \cdots \times n_D$ matrix. The linear subspace defined by
the following set of $r$ equations,
$$\sum_{\bi \in \bI} c_{\bi}h_{\bi , j}=0,$$
for all $0\leq j\leq r-1$, is an $[n_1\times n_2 \times \cdots
\times n_D , N -r]$ code. We say that $r$ is the redundancy of the
code.

Our goal in this paper is to handle $D$-dimensional errors from
one of the following types:

\begin{itemize}
\item Errors which don't spread more than one position around an
artificial center (which is the center of the error event).

\item Two errors in a cluster of some shape.
\end{itemize}
These clusters include the following types of bursts.

\begin{enumerate}
\item A $D$-dimensional 2-burst which corresponds to any two
adjacent positions that might be in error.

\item A $D$-dimensional 3-burst in which all the errors are on the
same line. Such an error corresponds to three positions of the
form $(i_1, \ldots, i_{j-1},i_j-1,i_{j+1},\ldots,i_D)$, $(i_1,
\ldots, i_{j-1},i_j,i_{j+1},\ldots,i_D)$, and $(i_1, \ldots,
i_{j-1},i_j+1,i_{j+1},\ldots,i_D)$ for some $j$, $1 \leq j \leq
D$.

\item A $D$-dimensional burst whose shape is a semi-cross with
arms of length one. Such a semi-cross has a center point at
$(i_1,i_2,\ldots,i_D)$ and includes all the points of the form
$(i_1, \ldots, i_{j-1},i_j+1,i_{j+1},\ldots,i_D)$, $1 \leq j \leq
D$.

\item A $D$-dimensional burst whose shape is a cross with arms of
length one. Such a cross has a center point at
$(i_1,i_2,\ldots,i_D)$ and includes all the points of the form
$(i_1, \ldots, i_{j-1},i_j-1,i_{j+1},\ldots,i_D)$ and all the
points of the form $(i_1, \ldots,
i_{j-1},i_j+1,i_{j+1},\ldots,i_D)$, $1 \leq j \leq D$.


\item Two errors inside a semi-cross or a cross with arms of
length $R$. A semi-cross with arms of length $R$ has a center
point at $(i_1,i_2,\ldots,i_D)$ and includes all the points of the
form $(i_1, \ldots, i_{j-1},i_j+\ell,i_{j+1},\ldots,i_D)$, $1 \leq
j \leq D$, $1 \leq \ell \leq R$. Similarly, a cross with arms of
length $R$ is defined. These errors are also related to two errors
inside a two-dimensional square with edges of length $R$.
\end{enumerate}

Why are we interested in crosses and semi-crosses? Errors are
likely to be spread within spheres to some limited radius. Crosses
and semi-crosses are types of spheres as described in~\cite{Gol69}
which are relatively simpler to handle than other spheres. These
spheres are also discussed extensively in the literature,
e.g.~\cite{GoWe,Ste,HaSt,StSz94}.
\section{Constructions with Low Redundancy}
\label{sec:optimal}

In this section we will handle two types of errors, 2-burst and
3-burst on a straight line. The number of possible patterns of
errors (excluding no errors) which can be confined to a
$D$-dimensional 2-burst is $D+1$ and to a $D$-dimensional 3-burst
on a line is $3D+1$. Hence, a lower bound of their redundancies is
$\text{log}~ ((D + 2) \cdot N)$ and $\text{log}~ ((3 D + 2) \cdot
N)$, respectively.
\subsection{Correction of 2-burst}

Assume that we have a $D$-dimensional array of size $n_1 \times
n_2 \times \cdots \times n_D$ on which we want to correct any
cluster of error that can be confined to a 2-burst.

\noindent {\bf Construction A:} Let $\alpha$ be a primitive
element in GF($2^m$) for $2^m-1 \geq \prod_{\ell=1}^D n_\ell$. Let
$d=\lceil \log D \rceil$ and $\bi=(i_1,i_2,\ldots,i_D)$, where $0
\leq i_\ell \leq n_\ell-1$. Let $A$ be a $d\times D$ matrix
containing distinct binary $d$-tuples as columns. We construct the
following $n_1 \times n_2 \times \cdots \times n_D \times (m+d+1)$
parity check matrix $H$:

$$h_{\bi}=
\left[\begin{array}{c}
1 \\
A\bi^T \mod 2 \\
\alpha^{\sum_{j=1}^D i_j (\prod_{\ell=j+1}^D n_\ell)}
\end{array}\right],
$$
for all $\bi=(i_1,i_2,\ldots,i_D)$, where $0 \leq i_\ell \leq
n_\ell-1$.
\begin{remark}
The matrix $A$ that we are using was also used in~\cite{ScEt}, but
the construction here is more flexible in its parameters. This is
a consequence from the way that we fold the elements of GF($2^m$)
into the parity-check matrix of the code. Moreover, we will see in
the sequel that this method will be of use for constructions of
larger bursts with at most two erroneous positions.
\end{remark}

In the decoding algorithm we assume that the error occurred can be
confined to a 2-burst. The syndrome received $v$ in the decoding
algorithm consists of three parts.
\begin{itemize}
\item The first bit determines the number of errors occurred.
Obviously if the syndrome is the all-zeroes vector than no errors
occurred. If the first bit of the syndrome is an {\it one} then
exactly one error occurred and its position is the position of $v$
in $H$. If the first bit of  a non-zero vector $v$ is a {\it zero}
then two errors occurred. Their position is determined by the
other $m+d$ entries of $v$.

\item The next $d$ bits determine the dimension in which the burst
occurred. There are $D$ dimensions and each column of the matrix
$A$ corresponds to a different dimension for two consecutive
errors. If the errors occurred in positions $\bi_1 =
(i_1,\ldots,i_D)$ and
$\bi_2=(i_1,\ldots,i_{j-1},i_j+1,i_{j+1},\ldots,i_D)$ then the
value of the $d$ bits, $(A\bi_1^T + A\bi_2^T)\bmod 2$, is the
$j$-th column of the matrix $A$.

\item The entries of the last $m$ rows of the matrix $H$ form the
folding of the first $\Pi_{\ell=1}^D n_{\ell}$ consecutive
elements of GF($2^m$). Given a dimension $\ell$ there exists an
integer $i(\ell)$ such that each two consecutive elements in
dimension $\ell$ have the form $\alpha^j , \alpha^{j+i(\ell)}$. It
is easy to verify that for $j_1 \neq j_2$ we have $\alpha^{j_1} +
\alpha^{{j_1}+i(\ell)} \neq \alpha^{j_2} +
\alpha^{{j_2}+i(\ell)}$. Thus, given the dimension of the burst of
size two, the last $m$ bits of $v$ can determine the two
consecutive positions of the burst.
\end{itemize}
\begin{theorem}
$~$
\begin{itemize}
\item The code constructed in Construction A can correct any error
pattern confined to a 2-burst in an $n_1 \times n_2 \times \cdots
\times n_D$ array codeword.

\item The code constructed by Construction A has redundancy which
is greater by at most two from the trivial lower bound on the
redundancy.
\end{itemize}
\end{theorem}
\begin{remark}
There are cases in which we can prove that the code of
Construction A is optimal.
\end{remark}
\subsection{3-burst on a line}

Next, we would like to consider correction of error patterns
confined to an arbitrary $D$-dimensional 3-burst. This appears to
be much more difficult than the 2-burst case. The main reason is
that the error can be spread on a two-dimensional subspace.
Therefore, we consider only the case of a $D$-dimensional cluster
of size three on an one-dimensional subspace, i.e. on a straight
line. In this case we can generalize Construction A.

\noindent {\bf Construction B:} Let $\alpha$ be a primitive
element in GF($2^m$) for $2^m-1 \geq \prod_{\ell=1}^D n_\ell$. Let
$B$ be a matrix of size $\lceil\text{log}~ (D+1)\rceil\times D$
which contains all binary representations of the integers between
$1$ and $D$ as its columns, such that the binary representation of
$D$ is the left most column, and the binary representation of $1$
is the right most column. The rows of $B$ are denoted by
$b_1,b_2,\ldots, b_{\lceil\text{log}~ (D+1)\rceil}$. Let $\beta$
be a primitive element in GF(4). By abuse of notation, if
$\bv=(v_1,v_2,\ldots,v_{\lceil\text{log}~ (D+1)\rceil})^T$ is a
column vector of length $\lceil\text{log}~ (D+1)\rceil$, we denote
$\beta^{\bv} = (\beta^{v_1},\beta^{v_2},\ldots,
\beta^{v_{\lceil\text{log}~ (D+1)\rceil}})^T$. We construct the
following $n_1 \times n_2 \times \cdots \times n_D \times
(m+2\lceil\text{log}~ (D+1)\rceil +2)$ parity check matrix $H^D$:
$$h_{\bi}^D= \left[\begin{array}{c}
1 \\
\left(\sum_{j=1}^Di_j\right) \bmod 2 \\
\beta^{B\bi^T} \\
\alpha^{\sum_{j=1}^D i_j (\prod_{\ell=j+1}^D n_\ell)}
\end{array}\right],$$
for all $\bi=(i_1,i_2,\ldots,i_D)$, where $0 \leq i_\ell \leq
n_\ell-1$. The multiplication $B\bi^T$ is taken over the integers
and the vector $\beta^{B\bi^T}$ consists of $\lceil\text{log}~
(D+1)\rceil$ vectors of length two, each one representing an
element in GF(4).
\begin{theorem}
$~$
\begin{itemize}
\item The code constructed in Construction B can correct any error
pattern confined to a 3-burst on a straight line in an $n_1 \times
n_2 \times \cdots \times n_D$ array codeword.

\item The code constructed by Construction B has excess redundancy
$2\lceil\text{log}~ (D+1)\rceil +2$ which is at most twice than
the trivial lower bound on the excess redundancy.
\end{itemize}
\end{theorem}
\section{The Coloring Method}
\label{sec:coloring}

The coloring method introduced in~\cite{EtYa} is an effective
method to handle multidimensional cluster errors. In the coloring
method we use $D$ one-dimensional auxiliary codes for our encoding
and decoding procedures. These codes are called component codes.
Each such a code $\cC_s$, $1 \leq s \leq D$, has length $\eta_s$
and we assign to it a coloring of the array codeword $\cA$.
Position $j$ of the code $\cC_s$ is the binary sum of all
positions in $\cA$ colored with color $j$ by the $s$-th coloring.
Assume that the size of the burst is $B$. The first code is a
$(B+\delta_1)$-burst-correcting code, $\delta_1 \geq 0$. This code
finds and corrects the shape of the error in the codeword of
$\cC_1$. The error that $\cC_1$ corrects can occur in a few
positions of the array codeword $\cA$. It might also have
different shapes in $\cA$, but the erroneous positions in $\cA$
have the colors of the positions which were in error in the
codeword of $\cC_1$. The $s$-th component code, $2 \leq s \leq D$,
is a $(B+\delta_s)$-burst-locator code, $\delta_s \geq \delta_1$
(usually, $\delta_s$ will be the same integer for all $2 \leq s
\leq D$). Burst-locator codes were discussed in~\cite{EtYa} and
are designed to find the location of a one-dimensional burst which
its shape is given up to a cyclic shift. Each of these codes,
$\cC_s$, provides additional information concerning the positions
of the errors, i.e., it reduces the sets of possible locations of
errors in $\cA$ as were found by $\cC_1, \cC_2 , \ldots,
\cC_{s-1}$. Finally, the last component code finds the actual
positions of the burst-error. To execute these tasks the colorings
should satisfy a few properties:
\begin{itemize}
\item ({\bf p.1}) For the $s$-th coloring, for each $s$, $1 \leq s
\leq D$, the colors inside a burst of the given shape are distinct
integers and the difference between the largest integer and the
smallest one is at most $B +\delta_s -1$.

\item ({\bf p.2}) Given the $D$ colorings and a color $\nu_s$, for
the $s$-th coloring, for each $1 \leq s \leq D$, there is at most
one position in the array which is colored with the colors
$(\nu_1, \nu_2 , \ldots, \nu_D )$.

\item ({\bf p.3}) Any two positions which are colored with the
same color by the first coloring, have colors which differ by a
multiple of $B + \delta_s$ by the $s$-th coloring, for each $s$,
$2 \leq s \leq D$.
\end{itemize}

The redundancy of the $D$-dimensional code is the sum of the
redundancies of the $D$ component codes. If we use a
$(B+\delta_s)$-burst-correcting code for the $s$-th component code
then this code does not need to satisfy ({\bf p.3}). The
disadvantage will be that the total redundancy of the
multidimensional code will increase. The advantage will be that we
will be more flexible in the parameters of the multidimensional
code since burst-locator codes are more rare than burst-correcting
codes.

Which codes can be used for the coloring method? We start with the
code for the first component code. The most efficient codes are
those constructed by Abdel-Ghaffar et al.~\cite{AMOT} for
correction of a $b$-burst. For a code of length $n$, the
redundancy of the code is $\lceil \text{log}~n \rceil +b-1$. The
main disadvantage of these codes is that their existence depends
on a sequence of conditions which are not easy to satisfy.

What about the locator codes? We can use locator codes derived
from the codes of Abdel-Ghaffar~\cite{AMOT,Abd88} as demonstrated
in~\cite{EtYa}. The redundancy of a locator code of length $n$ is
$\lceil \text{log}~ n \rceil$, i.e., it does not depend on the
length of the burst. But, these locator codes exist only for odd
burst length~\cite{EtYa}. Component codes with the parameters
$\eta_s$, $b$ and $D$, which satisfy ({\bf p.3}) are usually
difficult to find. Hence, if we want codes designed especially to
fit the parameters $\eta_s$, $b$ and $D$ we should compromise on
the redundancy of the component codes which will result in larger
redundancy of the multidimensional code. The best codes known for
this purpose are the Fire codes~\cite{Fire,Pat}. A Fire code of
length $n$ which corrects a $b$-burst has redundancy at most
$\lceil \text{log}~n \rceil +2b-1$.

It will be more convenient if each coloring is a linear function
of the coordinate indices, i.e., given a position $(i_1,i_2,
\ldots , i_D )$, its color for the $s$-th coloring is defined by
\begin{align*}
\sum_{k=1}^D \alpha_k^s i_k
\end{align*}
where $\alpha_k^s$ is a constant integer which depends on the
coloring $s$ and the shape of the $D$-dimensional cluster. Such a
coloring will be called a {\it linear coloring}. With a linear
coloring we associate a {\it coloring matrix} $A_D$, where $( A_D
)_{s,k} = \alpha_k^s$. It is easy to verify that property ({\bf
p.2}), is fulfilled for a linear coloring if and only if the
coloring matrix is invertible.

Now we will apply the coloring method on two types of errors,
semi-crosses and crosses with arms of length one. If we will try
to correct an error of either type by correcting a box-error which
inscribes it then the excess redundancy will be exponential in
$D$. The lower bound on the excess redundancy is linear in $D$ and
our code will have slightly larger excess redundancy. For
simplicity we will assume for the rest of this section that all
the edges of our array are equal to $n$. We use the notation $(
\times n)^D$ to denote $\underbrace{n\times \cdots \times
n}_{D\textrm{ times}}$.
\subsection{Semi-crosses with arms of length one}

We define the colorings by the coloring matrix, which is a
$D\times D$ matrix $A = \{a_{ij}\}_{1\leq i,j \leq D}$

$$
a_{1k}=k ,~~~ 1 \leq k \leq D ~.
$$
For each $s$, $2 \leq s \leq D$ we define

$$
a_{sk}=k , ~~~ 1 \leq k < s ~,
$$
$$
a_{sk}=k-D-1 , ~~~ s \leq k \leq D~.
$$
The $s$-th color, $1\leq s\leq D$, of position
$(i_1,i_2,\ldots,i_D)$ in the array is given by
$$c_{(i_1,i_2,\ldots,i_D)}^s = \sum_{k=1}^Da_{sk}i_k.$$

Using these colorings we obtain the following result. We present
here its proof in order to demonstrate how the coloring method
works. The proofs for all other colorings is similar as they all
satisfy properties ({\bf p.1}), ({\bf p.2}), and ({\bf p.3}).
\begin{theorem}
For any given even $D$, there exists a code which corrects any
$D$-dimensional error confined to a semi-cross burst with radius
one in an $( \times n)^D$ cube and its redundancy is at most
$\left\lceil\text{log}~ n^D\right\rceil + 2
D\left\lceil\text{log}~(D+1))\right\rceil + D$.
\end{theorem}
\begin{proof}
One can verify that the three coloring properties hold. Therefore,
given the set of erroneous colors by the first coloring, according
to property ({\bf p.3}) the shape of the burst in all other
colorings is known up to cyclic permutation. Therefore, for $2\leq
s\leq D$, the burst-locator code can find the locations of the
erroneous colors in the $s$-th coloring. Then, for each error in
the multidimensional array, its set of colors by each coloring is
known and according to property {\bf p.2} it is possible to find
the error location in the array.
\end{proof}

Better redundancy is obtained if we slightly change the coloring
and define a nonlinear coloring.
The $s$-th color, $1\leq s\leq D$, of position
$(i_1,i_2,\ldots,i_D)$ is given by
$$c_{(i_1,i_2,\ldots,i_D)}^s = \left(\sum_{k=1}^Da_{sk}i_k\right)\bmod (n(D+1)).$$

\noindent As a consequence we have the following theorems.
\begin{theorem}
For any given even $D$, there exists a code which corrects any
$D$-dimensional error confined to a semi-cross burst with radius
one in an $( \times n)^D$ cube and its redundancy is at most
$\left\lceil\log~ n^D\right\rceil +
D\left\lceil\text{log}~(D+1))\right\rceil + D$.
\end{theorem}

If we use the Fire codes~\cite{Fire,Pat} as locator codes we
obtain the following theorem.
\begin{theorem}
For any given $D$ and $n$, there exists a code which corrects any
$D$-dimensional error confined to a semi-cross burst with radius
one in an $( \times n)^D$ cube and its redundancy is at most
$\lceil \text{log}~ n^D \rceil + 2D^2 +
D\left\lceil\log~(D+1)\right\rceil +D$.
\end{theorem}
\subsection{Crosses with arms of length one}

We define the colorings by the coloring matrix, which is a
$D\times D$ matrix $A = \{a_{ij}\}_{1\leq i,j \leq D}$
$$a_{ij} = ij\bmod \left(2i(D-i+1)+1\right),$$
$$a_{ij} \in \{-i(D-i+1),\ldots,-1,0,1,\ldots, i(D-i+1)\}.$$ The
first color of position $(i_1,i_2,\ldots,i_D)$ in the array is
given by
$$c_{(i_1,i_2,\ldots,i_D)}^1 = \sum_{k=1}^Da_{1k}i_k.$$
The $s$-th color, $2\leq s\leq D$, of position
$(i_1,i_2,\ldots,i_D)$ in the array is given by
$$c_{(i_1,i_2,\ldots,i_D)}^s = \left(\sum_{k=1}^Da_{sk}i_k\right)\bmod (2s(D-s+1)n).$$
\begin{theorem}
\label{thm:red_Lee} There exists a code which corrects any
$D$-dimensional error confined to a Lee sphere burst with radius
one in an $( \times n)^D$ cube and its redundancy is at most
$\lceil \text{log}~n^D \rceil + 2D\left\lceil\text{log}~
D\right\rceil$.
\end{theorem}
\begin{theorem}
For any given $D$ and $n$, there exists a code which corrects any
$D$-dimensional Lee sphere burst with radius one in an $( \times
n)^D$ cube and its redundancy is at most $\lceil \text{log}~n^D
\rceil + 2D^2+ 2D\left\lceil\text{log}~ D\right\rceil$.
\end{theorem}
\begin{remark}
For specific values of $D$, i.e., when $2D+1$ is a prime number,
$11 \leq 2D+1 \leq 10000$, this construction can be slightly
improved.
\end{remark}
\section{Clusters with Limited Weight}
\label{sec:limit_weight}

When a certain area suffers from an error event we might expect
that not all the positions will be in error. Hence, it seems that
practically, we would expect to correct a cluster with a limited
weight. In this section we will consider the case where the weight
of the cluster is at most two.
\subsection{Semi-crosses}

We start by correcting an error with weight at most two in a
$D$-dimensional semi-cross with arms of length one.

\noindent {\bf Construction C:} Let $\alpha$ be a primitive
element in GF($2^m$) for $2^m-1 \geq \prod_{\ell=1}^D n_\ell$. Let
$d=\lceil \text{log}~ D \rceil$ and $\bi=(i_1,i_2,\ldots,i_D)$,
where $0 \leq i_\ell \leq n_\ell-1$. Let $\cH$ be a $(2d) \times
D$ parity-check matrix of a double-error correcting BCH code (or
its shortened code). We construct the following $n_1 \times n_2
\times \cdots \times n_D \times (m+2d+1)$ parity check matrix $H$:

$$h_{\bi}=
\left[\begin{array}{c}
1 \\
\cH \bi^T \bmod 2 \\
\alpha^{\sum_{j=1}^D i_j (\prod_{\ell=j+1}^D n_\ell)}
\end{array}\right],
$$
for all $\bi=(i_1,i_2,\ldots,i_D)$, where $0 \leq i_\ell \leq
n_\ell-1$.
\begin{theorem}
$~$
\begin{itemize}
\item The code constructed in Construction C can correct any error
of weight at most two inside a semi-cross with arms of length one in an
$n_1 \times n_2 \times \cdots \times n_D$ array codeword.

\item The code constructed by Construction C has redundancy which
is greater by at most two from the trivial lower bound on the
redundancy.
\end{itemize}
\end{theorem}
\begin{proof}
The first part of the Theorem is an immediate consequence from the
decoding procedure and the second part is easily verified. The
decoding is very similar to the one of Construction A. If the
received syndrome is the all zero vector then no error occurred.
The first bit of the syndrome indicates whether one or two errors
occurred. If the first bit is one, then one error occurred, and
its location can be found by the rest of the syndrome. Otherwise,
two errors occurred. The two errors can be of the form
$\bi_1=(i_1,\ldots,i_D)$ and
$\bi_2=(i_1,\ldots,i_{j-1},i_j+1,i_{j+1},i_D)$ or
$\bi_1=(i_1,\ldots,i_{k-1},i_k+1,i_{k+1},i_D)$ and
$\bi_2=(i_1,\ldots,i_{j-1},i_j+1,i_{j+1},i_D)$. In the first case,
the next $2d$ bits of the syndrome are the $j$-th column of the
matrix $\cH$, and in the second case the next $2d$ bits are the
sum of the $j$-th and $k$-th columns of the matrix $\cH$. Since
$\cH$ is a parity check matrix of a double error-correcting code
it is possible to distinguish between these cases and know the
shape of the error. Thus, as in Construction A, the last $m$ bits
of the syndrome indicate the location of the error.
\end{proof}

Construction C is generalized for a semi-cross with arms of length
$R$ with some extra redundancy in Construction D which follows.

\noindent {\bf Construction D:} Let $\alpha$ be a primitive
element in GF($2^m$) for $2^m-1 \geq \prod_{\ell=1}^D n_\ell$. Let
$t$ be the smallest integer such that $2^t-1 \geq 2RD$. Let $\cH$
be a $(4t) \times (2^t-1)$ parity-check matrix of a
four-error-correcting BCH code. Let $\Pi = \{ \hat{H}^1 ,
\hat{H}^2 , ..., \hat{H}^D \}$ be a set of disjoint subsets of columns of
size $2R$ from $\cH$, where $\hat{H}^i= [ \hat{h}^i_0, \ldots
,\hat{h}^i_{2R-1}]$. We construct the following $n_1 \times n_2
\times \cdots \times n_D \times (m+4t+1)$ parity-check matrix:
\vspace{-0.5ex}
$$h_{\bi}=
\left[\begin{array}{c}
1 \\
\sum_{\ell=1}^D \hat{h}^\ell_{i_\ell \bmod 2R} \bmod 2 \\
\alpha^{\sum_{j=1}^D i_j (\prod_{\ell=j+1}^D n_\ell)}
\end{array}\right],
$$
\vspace{-0.5ex}
for all $\bi=(i_1,i_2,\ldots,i_D)$, where $0 \leq i_\ell \leq
n_\ell-1$.
\begin{theorem}
$~$
\begin{itemize}
\item The code constructed in Construction D can correct any error
of weight at most two inside a semi-cross with arms of length $R$ in an
$n_1 \times n_2 \times \cdots \times n_D$ array codeword.

\item The code constructed by Construction D has excess redundancy
$4 \lceil \text{log}~ D + \text{log}~ R \rceil +5$ compared to a
trivial lower bound $2\text{log}~ D + 2\text{log}~ R -1$.
\end{itemize}
\end{theorem}
\begin{remark}
Two errors inside a two-dimensional square, of a $D$-dimensional
array, with edge of length $R$ can be viewed as an error with
weight two inside a semi-cross with arms of length $R-1$. Hence,
Construction D can be used to correct the related error.
\end{remark}\vspace{-0.5ex}
\subsection{Crosses}\vspace{-0.5ex}

The idea of correcting an error with weight two in a cross is a
modification of the one for a semi-cross. The two directions in
which the two errors occurred are revealed exactly as in the
semi-cross. The difference is that in the semi-cross the
directions are only positive, while in the cross they can be
either positive or negative. We will demonstrate how to solve the
problem in the case when the cross has arms with length one.
Similar solution is given for longer arms. It appears that we only
have to find whether the two directions have the same sign or not,
reducing the number of cases from four to two.

\noindent {\bf Construction E:} Let $\alpha$ be a primitive
element in GF($2^m$) for $2^m-1 \geq \prod_{\ell=1}^D n_\ell$. Let
$t$ be the smallest integer such that $2^t-1 \geq 4D$. Let $\cH$
be a $(4t) \times (2^t-1)$ parity-check matrix of a
four-error-correcting BCH code. Let $\Pi = \{ \hat{H}^1 ,
\hat{H}^2 , ..., \hat{H}^D \}$ be a set of disjoint subsets of columns of
size 4 from $\cH$, where $\hat{H}^i= [ \hat{h}^i_0, \hat{h}^i_1,
\hat{h}^i_2, \hat{h}^i_3]$. We construct the following $n_1 \times
n_2 \times \cdots \times n_D \times (m+4t+2)$ parity-check matrix:
\vspace{-0.5ex}
$$h_{\bi}=
\left[\begin{array}{c}
1 \\
\sum_{\ell=1}^D \hat{h}_{i_\ell \bmod 4}^\ell \bmod 2 \\
\lfloor \frac{\sum_{j=1}^D i_j}{2} \rfloor \bmod 2 \\
\alpha^{\sum_{j=1}^D i_j (\prod_{\ell=j+1}^D n_\ell)}
\end{array}\right],
$$\vspace{-0.5ex}
for all $\bi=(i_1,i_2,\ldots,i_D)$, where $0 \leq i_\ell \leq
n_\ell-1$.
\begin{theorem}
The code constructed in Construction E can correct any error of
weight at most two inside a cross with arms of length one in an $n_1
\times n_2 \times \cdots \times n_D$ array codeword.
\end{theorem}
\section{Conclusion}\vspace{-0.5ex}
\label{sec:conclude} We have considered various types of errors in
a $D$-dimensional array, where $D$ is a large integer. Given an
error pattern we have constructed codes for which the redundancy
and the excess redundancy are relatively small. The excess
redundancy is much smaller than the one obtained from a
construction which produces a code correcting a box-error which
contains the given cluster. The most immediate future goal will be
to construct $D$-dimensional codes which correct a cluster error
of size $b$ (whose shape is not a $D$-dimensional box), for large
$b$, with asymptotically optimal excess redundancy.
\section*{Acknowledgment}\vspace{-0.5ex}
The work of Eitan Yaakobi was supported in part by the University
of California Lab Fees Research Program, Award No.
09-LR-06-118620-SIEP. Tuvi Etzion was supported by the United
States-Israel Binational Science Foundation (BSF), Jerusalem,
Israel, under Grant No. 2006097.


\end{document}